\documentclass[paper]{ieice}
\usepackage{graphicx,xcolor}
\usepackage[fleqn]{amsmath}
\usepackage{newtxtext}
\usepackage[varg]{newtxmath}

\usepackage[caption=false,font=footnotesize]{subfig}
\usepackage{mathtools}
\usepackage{bm}
\usepackage{cite}
\usepackage{algorithm}
\usepackage{algorithmicx}
\usepackage{algpseudocode}
\algrenewcommand\algorithmicrequire{\textbf{Input:}}
\algrenewcommand\algorithmicensure{\textbf{Output:}}

\newtheorem{theorem}{Theorem}
\newtheorem{proof}{Proof}
\newtheorem{corollary}{Corollary}

\renewcommand{\(}{\left(}
\renewcommand{\)}{\right)}

\newcommand{\Tinterval}{T}
\newcommand{\Nveh}{N_\mathrm{v}}

\newcommand{\Vehicles}{\mathcal{V}}
\newcommand{\Link}{\mathcal{L}}

\newcommand{\ds}{\bm{d}}
\newcommand{\op}{t}
\newcommand{\IsOp}{\bm{t}}
\newcommand{\Operation}{\mathcal{T}}
\newcommand{\weight}{W}
\newcommand{\NumSlot}{\tau_\mathrm{max}}
\newcommand{\tend}{\tau_\mathrm{end}}
\newcommand{\dist}{\mathit{dist}}
\newcommand{\centerpos}{p_\mathrm{c}}
\newcommand{\Loss}{\mathit{LOSS}}
\newcommand{\Conflict}{\mathcal{C}}
\newcommand{\SchedGraph}[1][]{\bm{G}_{\mathrm{s}#1}}
\newcommand{\VehGraph}[1][]{\bm{G}_{\mathrm{v}#1}}

\newcommand{\NormalCov}[1][]{\hat{S}_{#1}}
\newcommand{\EachCov}{S}
\newcommand{\Region}{R}
\newcommand{\RegionAll}{R_\mathrm{all}}
\newcommand{\SensorRange}{r_\mathrm{s}}
\newcommand{\BandWidth}{B}
\newcommand{\Noise}{N}
\newcommand{\GainTx}{G_\mathrm{t}}
\newcommand{\GainRx}{G_\mathrm{r}}

\newcommand{\SINR}{\mathit{SINR}}
\newcommand{\sinr}{\mathit{sinr}}
\newcommand{\Rate}{\mathit{Rate}}
\newcommand{\Itfr}{I}
\newcommand{\PowerTx}{P_\mathrm{t}}
\newcommand{\PowerRx}[1][]{P_{\mathrm{r}#1}}
\newcommand{\Threshold}{\theta}
\newcommand{\Neighbor}{\mathcal{N}_{\VehGraph}}
\newcommand{\irx}{i_\mathrm{r}}
\newcommand{\itx}{i_\mathrm{t}}
\newcommand{\jrx}{j_\mathrm{r}}
\newcommand{\jtx}{j_\mathrm{t}}
\newcommand{\krx}{k_\mathrm{r}}
\newcommand{\ktx}{k_\mathrm{t}}
\newcommand{\IntVDist}{l_\mathrm{avg}}
\newcommand{\DistNum}[2]{\(#1\,\mathrm{m},#2\)}

\usepackage{color}

\field{D}
\SpecialSection{the Architectures, Protocols, and Applications for the Future Internet}
\title{Concurrent Transmission Scheduling for Perceptual Data Sharing \\ in mmWave Vehicular Networks}
\titlenote{A part of this paper was presented at IEEE CAVS 2018.}
\authorlist{%
\breakauthorline{2}
\authorentry{Akihito Taya}{s}{kyoto-u}\MembershipNumber{1102117}
\authorentry[nishio@i.kyoto-u.ac.jp]{Takayuki Nishio}{m}{kyoto-u}\MembershipNumber{1002899}
\authorentry{Masahiro Morikura}{f}{kyoto-u}\MembershipNumber{8110692}
\authorentry{Koji Yamamoto}{e}{kyoto-u}\MembershipNumber{0212512}
}
\affiliate[affiliate label]{The author is with the Graduate School of Informatics, Kyoto University, Yoshida-honmachi, Sakyo-ku, Kyoto, 606--8501 Japan.}

\received{20xx}{1}{1}
\revised{20xx}{1}{1}

\begin{document}
\maketitle
\begin{summary}
  Sharing perceptual data (e.g., camera and LiDAR data) with other vehicles
  enhances the traffic safety of autonomous vehicles
  because it helps vehicles locate other vehicles and pedestrians in their blind spots.
  Such safety applications require high throughput and short delay,
  which cannot be achieved by conventional microwave vehicular communication systems.
  Therefore, millimeter-wave (mmWave) communications are considered to be a key technology
  for sharing perceptual data because of their wide bandwidth.
  One of the challenges of data sharing in mmWave communications is broadcasting
  because narrow-beam directional antennas are used to obtain high gain.
  Because many vehicles should share their perceptual data to others within a short time frame
  in order to enlarge the areas that can be perceived based on shared perceptual data,
  an efficient scheduling for concurrent transmission that improves spatial reuse is required
  for perceptual data sharing.
  This paper proposes a data sharing algorithm that employs a graph-based concurrent transmission scheduling.
  The proposed algorithm realizes concurrent transmission to improve spatial reuse
  by designing a rule that is utilized to determine
  if the two pairs of transmitters and receivers interfere with each other
  by considering the radio propagation characteristics of narrow-beam antennas.
  A prioritization method that considers the geographical information in perceptual data is also designed
  to enlarge perceivable areas
  in situations where data sharing time is limited and not all data can be shared.
  Simulation results demonstrate that the proposed algorithm
  doubles the area of the cooperatively perceivable region
  compared with a conventional algorithm that does not consider mmWave communications
  because the proposed algorithm achieves high-throughput transmission
  by improving spatial reuse.
  The prioritization also enlarges the perceivable region by a maximum of 20\%.
\end{summary}

\section{Introduction}
Millimeter-wave (mmWave) vehicular adhoc networks (VANETs)
are expected to be an enabler of numerous safety applications for autonomous vehicles
that require high-throughput transmission capability
\cite{MmWaveVanetSurvey,eband,PathLossPrediction,perfecto2017millimeter,wu2017cooperative}.
As vehicles become increasingly automated,
the number of sensors equipped on vehicles increases
and an increasingly massive amount of data are generated while driving.
Sharing these sensor data, such as camera and LiDAR data,
would help extend a vehicle's perceptual range to cover its blind spots or locate hidden objects.
However, a sufficient data rate for sharing sensor data cannot be provided
by currently standardized vehicular communication systems
(e.g., IEEE 802.11p/dedicated short range communications (DSRC)
and cellular vehicle-to-everything (C-V2X),
standardized in the third generation partnership project (3GPP) Release 14 \cite{molina2017lte})
because of their limited bandwidth.
Therefore, mmWave communications,
which provide high-throughput communication
by leveraging huge bandwidth and efficient spatial reuse,
have been attracting much attention for vehicular communications.

One of the most important traffic safety applications facilitated by mmWave communications is cooperative perception,
which enables autonomous vehicles to perceive their blind spots
by sharing perceptual data, such as camera, LiDAR, and radar data, with other vehicles.
For example, see-through systems provide following vehicles with front views of the leader of platooning vehicles
and bird's-eye-view systems generate top views of surrounding areas by aggregating perceptual data of multiple vehicles
\cite{kim2015multivehicle,li2011multi}.
Such techniques are particularly important at intersections with poor visibility to avoid car crash.
By sharing information regarding their surroundings,
the region that autonomous vehicles can perceive is enlarged based on the shared information.
Computer vision systems enable vehicles to recognize other vehicles, pedestrians, and traffic signs,
even if they cannot be seen directly because buildings or other obstacles block the line of sight.
To cover the entire area surrounding an intersection,
vehicles near the intersection should send their massive data to the other vehicles
within a short period, in particular 100\,ms for safety applications \cite{VehNetworking}.
For example, assume 20 vehicles attempt to share compressed camera images within 100\,ms.
The image sizes range from 1--9\,Mbit because they are generated at rates of 10--90\,Mbit/s \cite{MmWaveVANET}.
Therefore, 20--180\,Mbit of data must be transmitted within 100\,ms by 20 vehicles,
meaning each datum must be transmitted at a rate of 0.2--1.8\,Gbit/s.
Such a high-throughput system is difficult to be realized
by DSRC or C-V2X owing to their limited bandwidth.

Although mmWave communications enable high-throughput transmission,
it is difficult to broadcast data to all vehicles compared with microwave communications
because few vehicles can receive transmitted signals
because of narrow-beam directional antennas and severe attenuation by the blockage effect.
Therefore, an efficient mechanism to share perceptual data in mmWave multihop networks should be developed.
As mentioned above, vehicles are required to share perceptual data
and obtain data of as wide region as possible within 100\,ms.
To meet these requirements,
concurrent transmission and routing with cached data are promising approaches.
Concurrent transmission,
where many transmitters send data to different receivers at the same time,
promotes efficient spatial reuse,
which is realized by leveraging antenna directionality and high attenuation.
On the other hand,
routing using cached data reduces redundant transmissions for data sharing in multihop networks
because each datum is requested to be sent to many different vehicles.
In multihop networks, if relay vehicles store the forwarded data,
the source vehicles do not need to transmit the same data many times.
Leveraging the geographical information in perceptual data also helps
to enlarge perceivable regions.

There have been a few studies on concurrent transmissions in mmWave VANET.
For example, \cite{perfecto2017millimeter} proposed a beam-width-controlling scheme to reduce beam-alignment delay
by considering signal-to-interference plus noise power ratio (SINR).
Most concurrent transmission protocols for mmWave communications
are found not in VANETs, but in wireless sensor networks (WSNs)
\cite{qiao2012stdma,niu2015blockage,wang2014throughput}.
However, such protocols do not adopt data caching because their objectives are not data sharing.
In data sharing, the same data are sent from the source vehicles to different vehicles
and thus, the same data might be transmitted redundantly without data caching.
Additionally, their algorithms do not consider the geographical information in transmitted data.
There have been many studies on data dissemination methods for DSRC-based VANETs,
some of which utilize the geographical information in disseminated data.
\cite{wischhof2005information,bronsted2006specification}
proposed the data aggregation of the geographical information
to suppress redundant data broadcasts.
\cite{yamada2017data} proposed controlling the frequency of broadcasting.
However, these studies did not discuss concurrent transmission
or multihop routing with directional antennas.

Concurrent dissemination with data caching was proposed in \cite{coopDataSched},
where the authors presented a system to realize
a road-side-unit (RSU)-controlled concurrent dissemination by two communication mode:
vehicle-to-infrastructure (V2I) and vehicle-to-vehicle (V2V) communications.
\cite{coopDataSched} proposed a graph-based algorithm,
where potential transmissions (from which, to which, and which data should be transmitted)
and their conflicts (e.g. half duplex and interference constraints)
are represented as a graph,
i.e., two transmissions are connected when they cannot be operated at the same time.
Each vertex has weight that represents the priority of receiver vehicles.
Then, the optimal concurrent transmission schedule for multihop dissemination
can be obtained by solving the maximum weighted independent set (MWIS) problem on the graph.
Although the MWIS problem is one of the NP-hard problems,
a greedy algorithm with a performance guarantee to maximize the total vertex weights can be utilized.
However, \cite{coopDataSched} does not assume mmWave communications
and thus, it cannot be used directly for mmWave communications.
There is also room to utilize the geographical information in the transmitted data for cooperative perceptions.

In this paper,
we propose a mmWave data sharing algorithm
where vehicles share perceptual data with each other
and enlarge the perceivable regions.
The proposed algorithm is based on \cite{coopDataSched}
to realize concurrent transmission with data caching.
Because the algorithm in \cite{coopDataSched} optimizes concurrent transmissions
considering not only pair selection of the transmitter and receiver
but also which data to transmit among the currently cached data,
it effectively reduces redundant transmission in data sharing,
where the same data are transmitted to different vehicles.
However, the original algorithm is based on microwave communications,
meaning it must be modified for mmWave communications.
\cite{coopDataSched} designed a conflict rule,
which is utilized to decide which pair of transmission vertices of the graph should be connected,
considering radio interference among omnidirectional antennas.
We design a new rule for mmWave communications
by estimating interference among narrow-beam directional antennas,
which are utilized for mmWave communications to obtain high gain.
Because the conflict rule should be defined between two transmissions,
we develop an interference approximation scheme
that can calculate the interference between two transmissions
without summing all possible interferences.
By using the newly designed rule,
near-optimal concurrent transmission in mmWave networks can be realized.

We also design a prioritization method in order to enlarge the perceivable region
for situations where data sharing time is limited.
Although \cite{coopDataSched} gave high priority to receiver vehicles
that soon run out of the service area of the RSU,
such a prioritization does not fit for cooperative perceptions at an intersection.
We give high priority to data corresponding to regions far from an intersection
to enlarge the perceivable area based on shared data
because regions near the center of the intersection are covered by many vehicles,
meaning it is desirable to transmit data far from the intersection.
Such a prioritization scheme can be realized by customizing the weight function of the MWIS problem.

The main contributions of this paper are summarized as follows:
(1) We propose a data sharing algorithm for cooperative perception,
which improves spatial reuse by considering interference among narrow-beam directional antennas
and increases the perceivable region
by prioritizing the data to be forwarded based on geographical information,
even if not all data can be collected.
In order to realize concurrent transmission,
we employ the algorithm presented in \cite{coopDataSched}.
(2) We prove that
if the data sharing time is sufficiently long and the vehicular network is represented as a connected graph,
the proposed data sharing algorithm guarantees
that all data are shared with all vehicles.

\section{Related Works} \label{sec:related}
Data sharing for cooperative perceptions should
achieve a large perceivable area within a short period,
in particular 100\,ms for safety applications.
Key techniques to meet this requirement are
concurrent transmission with directional antennas for improving system throughput,
efficient routing with cached data for reducing redundant transmission,
and leveraging geographical information in transmitted data.

Dissemination algorithms with directional antennas for VANETs have been studied by many researchers.
\cite{li2012collaborative} presented theoretical analysis of content dissemination time
in vehicular networks with directional antennas
and demonstrated that directional antennas accelerate content propagation.
\cite{li2007distance} proposed a broadcast protocol for directional antennas in VANETs.
In this protocol, the furthest receiver forwards data packets along road segments
and a directional repeater forwards the data in multiple directions at intersections.
In contrast to the protocol in \cite{li2007distance},
which considers the positions of transmitters,
our algorithm considers the positions where data are obtained
to achieve a large perceivable region.

Dissemination algorithms for local information were proposed
in \cite{wischhof2005information,bronsted2006specification,yamada2017data}.
In \cite{wischhof2005information},
a scalable dissemination protocol,
called segment-oriented data abstraction and dissemination (SODAD),
and its application, self-organizing traffic-information system (SOTIS), were proposed.
SOTIS is a mechanism for gathering traffic information sensed by vehicles.
It aggregates the received traffic information from road segments
and sends only up-to-date information to vehicles.
In \cite{bronsted2006specification},
Zone Flooding and Zone Diffusion were proposed to suppress redundant data broadcasting.
In Zone Flooding, only vehicles in a flooding zone forward received packets.
Zone Diffusion is a data aggregation method considering geographical information,
where vehicles merge road environment data as it is received
and broadcast only merged data.
\cite{yamada2017data} proposed controlling the frequency of information broadcasting
and selecting the data to send to reduce communication traffic.
Although these studies considered the geographical information in each datum,
they did not focus on concurrent transmission.

The authors of \cite{perfecto2017millimeter} proposed
vehicle pairing and beam-width controlling for mmWave VANETs.
In the protocol in \cite{perfecto2017millimeter},
pairs of transmitters and receivers are selected based on matching theory
and beam widths are determined via particle swarm optimization.
This protocol successfully improves throughput and reduces delay by considering SINR.
Other concurrent transmission methods for mmWave communications have been proposed for WSN, rather than VANETs
\cite{qiao2012stdma,niu2015blockage,wang2014throughput}.
The authors of \cite{qiao2012stdma} formulated the concurrent transmission scheduling problem
as an optimization problem to maximize the number of flows
to satisfy the quality-of-service requirements of each flow.
In \cite{niu2015blockage},
relay selection and spatial reuse were jointly optimized to improve network throughput
and a blockage robust algorithm was proposed.
The authors of \cite{wang2014throughput} minimized transmission time
by solving an optimization problem.
Although these algorithms for concurrent transmissions
presented in \cite{perfecto2017millimeter,qiao2012stdma,niu2015blockage,wang2014throughput}
achieved efficient spatial reuse,
redundant data were transmitted because their primary objective was not data sharing
and thus, they did not consider situations where the same data are sent to different receivers.
Additionally, they did not consider geographical information.

The authors of \cite{coopDataSched} proposed an RSU-controlled scheduling
that maximizes system throughput in hybrid V2I/V2V communications.
This algorithm realizes concurrent dissemination based on the graph theory.
It also adopts a data caching mechanism.
The algorithm proposed in \cite{coopDataSched}
generates graphs for dissemination scheduling,
where the set of vertices represents potential transmissions
consisting of a transmitter, receiver, and data,
and the set of edges represents pairs of transmissions that cannot be performed at the same time.
The authors of \cite{coopDataSched} proved that
optimal scheduling can be obtained
by solving the MWIS problem for a generated graph.
However, because the algorithm in \cite{coopDataSched} assumes omnidirectional antennas,
interference calculations must be extended for mmWave communications,
where narrow-beam directional antennas are utilized.
Additionally, there is still room to improve the efficiency of data transmissions for cooperative perception
by leveraging the geographical information in perceptual data.
Thus, a data sharing algorithm in mmWave vehicular networks
that increases perceivable regions should be developed for traffic safety,
especially when data sharing time is limited.

\begin{figure}[t] \centering
  \includegraphics[width=0.45\textwidth]{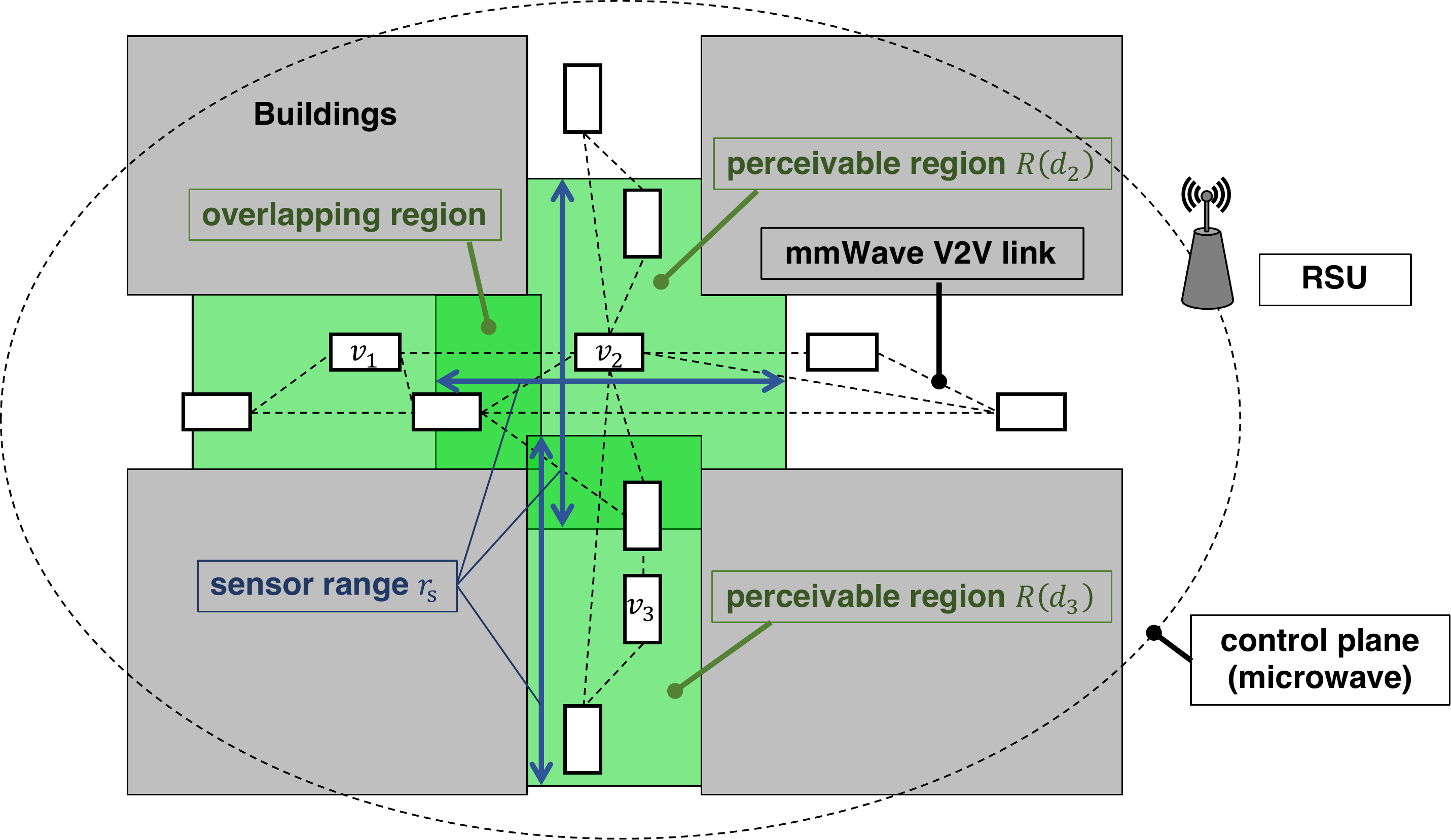}
  \caption{System model (Top view).}
  \label{fig:systemmodel}
\end{figure}

\section{System Model} \label{sec:system}
Figure~\ref{fig:systemmodel} shows our system model.
At an intersection,
there are vehicles equipped with mmWave communication devices for data transmission via V2V channels
and microwave communication devices for control signal transmission via V2I channels.
Vehicles participating in cooperative perception are selected
among vehicles within tens of meters from the center of the intersection considering stopping distance.
The number of participants is also limited to $\Nveh$ vehicles
because it is difficult to complete data sharing owing to the time limit
when the number of participants is large.
The vehicles perceive the surrounding environment
utilizing their sensors, such as LiDARs or cameras.
We assume that the vehicle sensors cover a surrounding rectangular region (on road segments)
or cross-shaped region (at the intersection),
bounded by the buildings along the roads
and their sensor range $\SensorRange$.
The data generated by vehicle $v_i$ is denoted as $d_i$.
We assume the sizes of $d_i$ are approximately the same among vehicles for simplicity.
The vehicles share the data with each other
to obtain information regarding the intersection
and then perform cooperative perception.

Data are transmitted through mmWave V2V channels
to reduce the pressure on V2I channels.
However, control signals, which must be broadcasted to all vehicles,
are transmitted through microwave V2I channels.
We assume there is an RSU (or an eNodeB) that covers all vehicles near the intersection on the microwave channel
and performs scheduling based on vehicle positions and mmWave V2V link topology.
While a large amount of sensor data are transmitted over the mmWave V2V channels,
control signals and position information, which are relatively small,
can be broadcasted by the RSU utilizing DSRC or C-V2X.

Figure~\ref{fig:time} shows the time frame for data sharing.
The vehicles perform sensing at an interval of $\Tinterval$ and generate $d_i$.
The data update interval consists of
the scheduling period and data sharing period.
In the scheduling period, data sharing scheduling is determined by the RSU.
Vehicle position information
obtained from global positioning system (GPS)
is sent to the RSU,
which then estimates the mmWave connectivity between vehicles
and determines the preferred data to be shared.

\begin{figure}[t] \centering
  \includegraphics[width=0.48\textwidth]{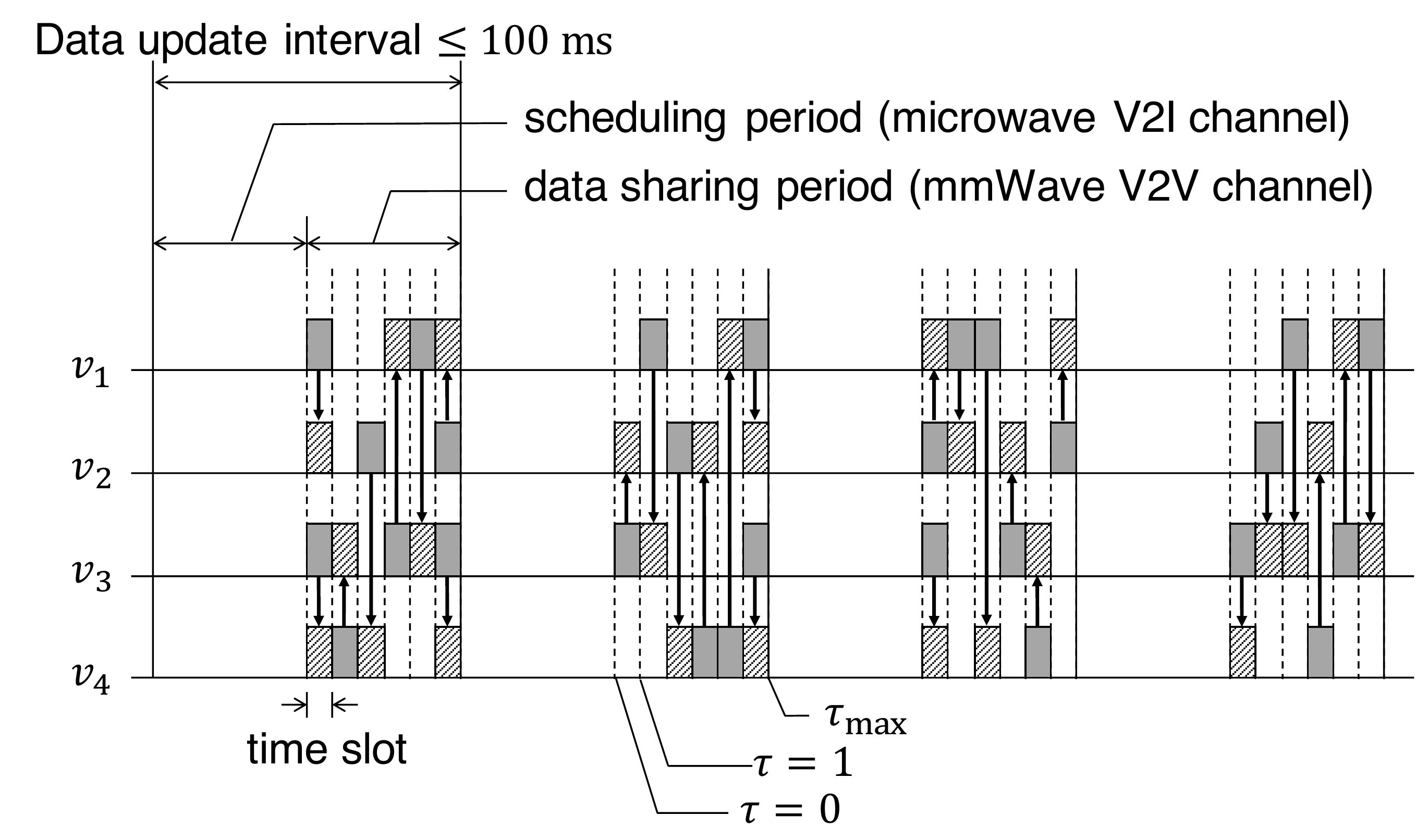}
  \caption{Time frame for data sharing.
           All vehicles generate their perceptual data at the beginning of each data update interval.
           They transmit their data during the data sharing period.}
  \label{fig:time}
\end{figure}

In the data sharing period,
vehicles share their data through mmWave V2V channels.
The data sharing period consists of $\NumSlot$ time slots,
each of which is sufficiently long to transmit one datum.
Let $\tau\in\{0,1,\dots,\NumSlot\}$ denote the index of a time slot.
$\NumSlot$ is limited by the transmitted data volume and data rate.

By sharing data $d_i$, the perceivable region is enlarged.
Let $\ds_{i,\tau}$ and $\Region_{i,\tau}$ denote
the dataset possessed by vehicle $v_i$
and the perceivable region of the dataset $\ds_{i,\tau}$, respectively.
$\Region_{i,\tau}$ is defined as
$\Region_{i,\tau} \coloneqq \bigcup_{d \in \ds_{i,\tau}}\Region(d)$,
where $\Region(d)$ denotes the perceivable region of $d$
(i.e., the region covered by the sensor of a single vehicle).
At the beginning of the data sharing period,
the datasets are initialized as $\ds_{i,0} \leftarrow \{d_i\}$.
When vehicle $v_i$ transmits $d_k \in \ds_{i,\tau}$ to vehicle $v_j$,
$v_j$ updates its dataset $\ds_{j,\tau}$ as follows:
\begin{align}
  \ds_{j,\tau+1} \leftarrow \ds_{j,\tau} \cup \{d_k\}. \label{eq:updatedataset}
\end{align}
Subsequently, the area of region $\Region_{j,\tau}$ is enlarged.
We evaluate system performance
based on the normalized perceivable area, which is defined as follows:
\begin{align}
  \NormalCov[\tau] &\coloneqq \EachCov(\Region_{i,\tau}) / \EachCov(\RegionAll), \label{eq:nrm}
\end{align}
where $\EachCov(\Region)$ and $\RegionAll$ denote
the area of region $\Region$
and the area covered by all data,
defined as $\RegionAll \coloneqq \bigcup_{i=1}^{\Nveh}\Region(d_i)$, respectively.

At the end of the data sharing period, the vehicles regenerate $d_i$ by sensing.
Then, the RSU collects vehicle position information and
determines scheduling for sharing new perceptual data in the following scheduling period.

\section{Data Sharing Algorithm} \label{seq:proposed}
In the scheduling period, the RSU selects
transmitters, receivers, and data to be transmitted during each time slot.
First, the RSU constructs a vehicular network graph
that represents the network topology of the mmWave vehicular network
by estimating the connectivity between vehicles
based on their positions and a propagation loss model.
Second, a graph that is utilized to determine concurrent transmission behavior
is constructed from the vehicular network graph for each time slot.
Because the vertices of the graph represent transmissions,
each of which consists of a transmitter, receiver, and data to be transmitted,
and the edges of the graph represent conflicts between two transmissions,
independent sets in the graph represent sets of transmissions that do not conflict with each other.
Therefore, by solving the MWIS problem for the graph,
which we call a scheduling graph,
the optimal concurrent transmission can be found for each time slot.

Although our algorithm is based on that proposed in \cite{coopDataSched},
our system model is quite different from that in \cite{coopDataSched}.
First, we consider a short period (i.e., 100\,ms),
while long-span dissemination was discussed in \cite{coopDataSched}.
Because \cite{coopDataSched} designed a prioritization based on vehicle mobility over a long period,
we modify the prioritization design to enlarge perceivable regions within a short period.
Second, our objective is to share data generated by vehicles with each other,
while \cite{coopDataSched} assumed that each vehicle requests data that is stored in the RSU.
We prove that data sharing can be completed by our algorithm
if the vehicular network graph is connected and there are enough time slots.
Third, we utilize mmWave communications for data transmission
and thus, we redesign how to construct the scheduling graph.
Especially, conflict rules between two potential transmissions
are modified to reflect the mmWave propagation characteristics.
Finally, data are transmitted through V2V channels in our system
to reduce the pressure on V2I channels,
while \cite{coopDataSched} utilized both V2I and V2V channels for data transmission.
The following subsections describe the details of constructing
a vehicular network graph and scheduling graph.

\subsection{Vehicular Network Graphs}
The RSU estimates the connectivity between each pair of vehicles
and defines the vehicular network graph $\VehGraph$ as follows:
\begin{align}
  \VehGraph &\coloneqq \(\Vehicles, \Link\), \label{eq:vehgraph} \\
  \Link &\coloneqq \{\{v_i, v_j\} \mid v_i, v_j \in \Vehicles, \Loss(v_i, v_j) \leq \Threshold\}, \label{eq:link}
\end{align}
where $\Vehicles$, $\Link$, $\Loss(v_i,v_j)$, and $\Threshold$ denote
the set of vehicles,
set of vehicle connections,
mmWave propagation loss between $v_i$ and $v_j$,
and a threshold that indicates that mmWave communications are possible, respectively.
The mmWave propagation loss can be obtained from the path loss models proposed in \cite{PathLossPrediction}.
The authors of \cite{PathLossPrediction} measured the propagation loss of 60-GHz mmWave channels
when there were one, two, or three vehicles between the transmitter and receiver.
For scenarios with more than three blockers, \cite{MmWaveVanetSurvey} provided an extension to the path loss model.
Another approach for predicting mmWave propagation loss was proposed in \cite{RNNbasedRSSPred},
where the authors predicted received signal power based on perceptual data.
The threshold $\Threshold$ is calculated based on the Shannon capacity as follows:
\begin{align}
  &\BandWidth \log_2 \(1 + \frac{\PowerTx \GainTx \GainRx / \Loss(v_i, v_j)}{\BandWidth \Noise}\) \geq \Rate, \\
  &\Threshold \coloneqq \frac{\PowerTx \GainTx \GainRx}{\BandWidth \Noise \( 2^{\Rate/\BandWidth} - 1 \)},
\end{align}
where $\BandWidth$, $\PowerTx$, $\GainTx$, $\GainRx$, $\Noise$, and $\Rate$ denote
the bandwidth, transmission power, transmitter and receiver antenna gain,
thermal noise power spectral density, and rate requirements, respectively.
When calculating the vehicle connectivity,
the antenna directions of the transmitter and receiver point at each other.
We also assume that $\VehGraph$ does not change within the data update interval
because the interval is very short (less than 100\,ms),
meaning the mobility of the vehicles is negligible.

\begin{algorithm}[t]
  \caption{Constructing a scheduling graph} \label{alg:sched_graph}
  \begin{algorithmic}[1]
    \Require {$\VehGraph=\(\Vehicles, \Link\), \ds_{i,\tau}$}
    \Ensure {$\SchedGraph[,\tau]$}
    \State {Initialize $\Operation_\tau \leftarrow \emptyset, \Conflict_\tau \leftarrow \emptyset$}
    \ForAll {$v_i \in \Vehicles$}
      \ForAll {$v_j \text{ in neighbors of } v_i$}
        \ForAll {$d_k \in \ds_{i,\tau}$}
          \If {$d_k \notin \ds_{j,\tau}$}
            \State $\Operation_\tau \leftarrow \Operation_\tau \cup \{\op_{ijk}\}$
          \EndIf
        \EndFor
      \EndFor
    \EndFor
    \ForAll {$\op_{ijk}, \op_{i'j'k'} \in \Operation_\tau$}
      \If {$\op_{ijk}$ and $\op_{i'j'k'}$ conflict with each other}
        \State {$\Conflict_\tau \leftarrow \Conflict_\tau \cup \{\{\op_{ijk}, \op_{i'j'k'}\}\}$}
      \EndIf
    \EndFor
    \State {$\SchedGraph[,\tau] \leftarrow (\Operation_\tau, \Conflict_\tau, \weight)$}
    \State \Return $\SchedGraph[,\tau]$
  \end{algorithmic}
\end{algorithm}

\begin{algorithm}[t]
  \caption{Data sharing scheduling} \label{alg:scheduling}
  \begin{algorithmic}[1]
    \Require {$\VehGraph=\(\Vehicles, \Link\)$.}
    \State Initialize $\ds_i \leftarrow \{d_i\}$ for all $i$
    \State Obtain $\SchedGraph[,0]$ from Algorithm~\ref{alg:sched_graph} with $\VehGraph, \ds_{i,0}$
    \State $\tau \leftarrow 0$
    \While {$\Operation_\tau \neq \emptyset \land \tau \le \NumSlot$}
      \State $\IsOp_\tau \leftarrow \text{ MWIS of } \SchedGraph[,\tau]$
      \State Perform $\IsOp_\tau$ and update $\ds_{i,\tau+1}$
      \State Obtain $\SchedGraph[,\tau+1]$ from Algorithm~\ref{alg:sched_graph} with $\VehGraph, \ds_{i,\tau+1}$
      \State $\tau \leftarrow \tau+1$
    \EndWhile
  \end{algorithmic}
\end{algorithm}

\subsection{Scheduling Graph and Data Sharing Scheduling}
\begin{figure}[t]
  \centering
  \subfloat[Vehicular network graph $\VehGraph$]{
    \includegraphics[width=0.28\textwidth]{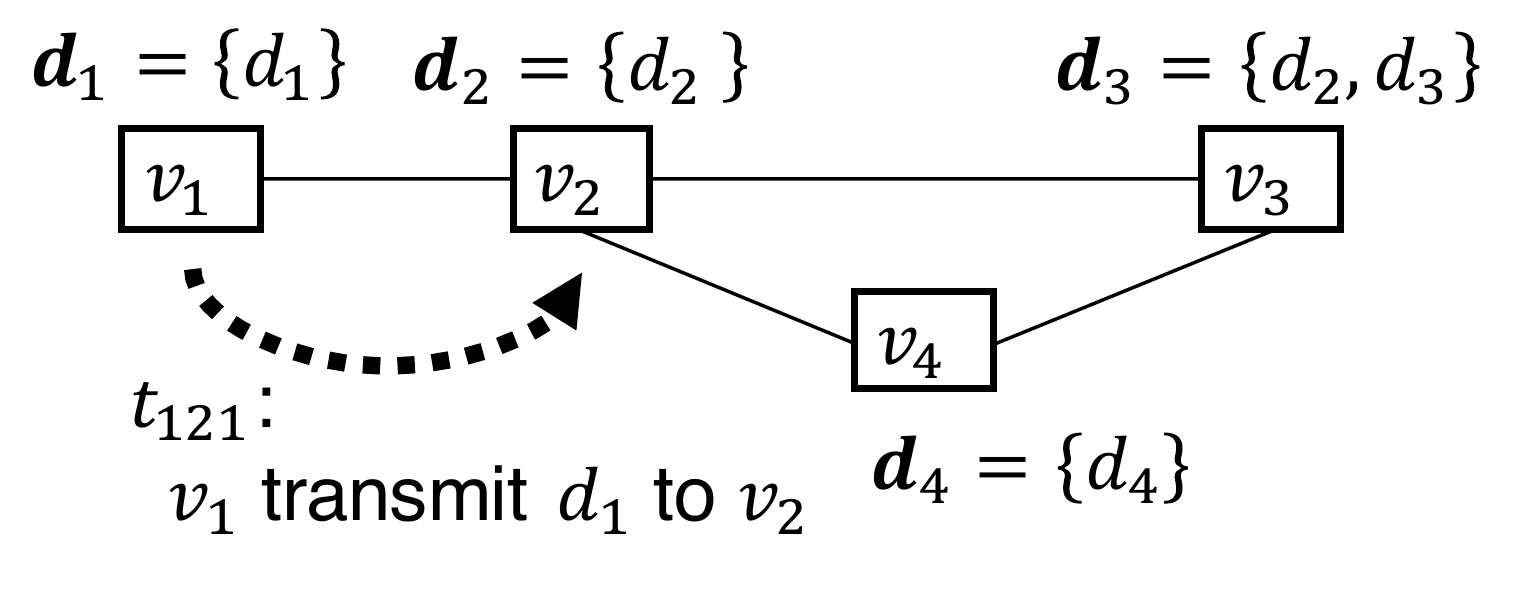}{}
    \label{fig:veh_graph}
  }
  \hfil
  \subfloat[Scheduling graph $\SchedGraph$]{
    \includegraphics[width=0.17\textwidth]{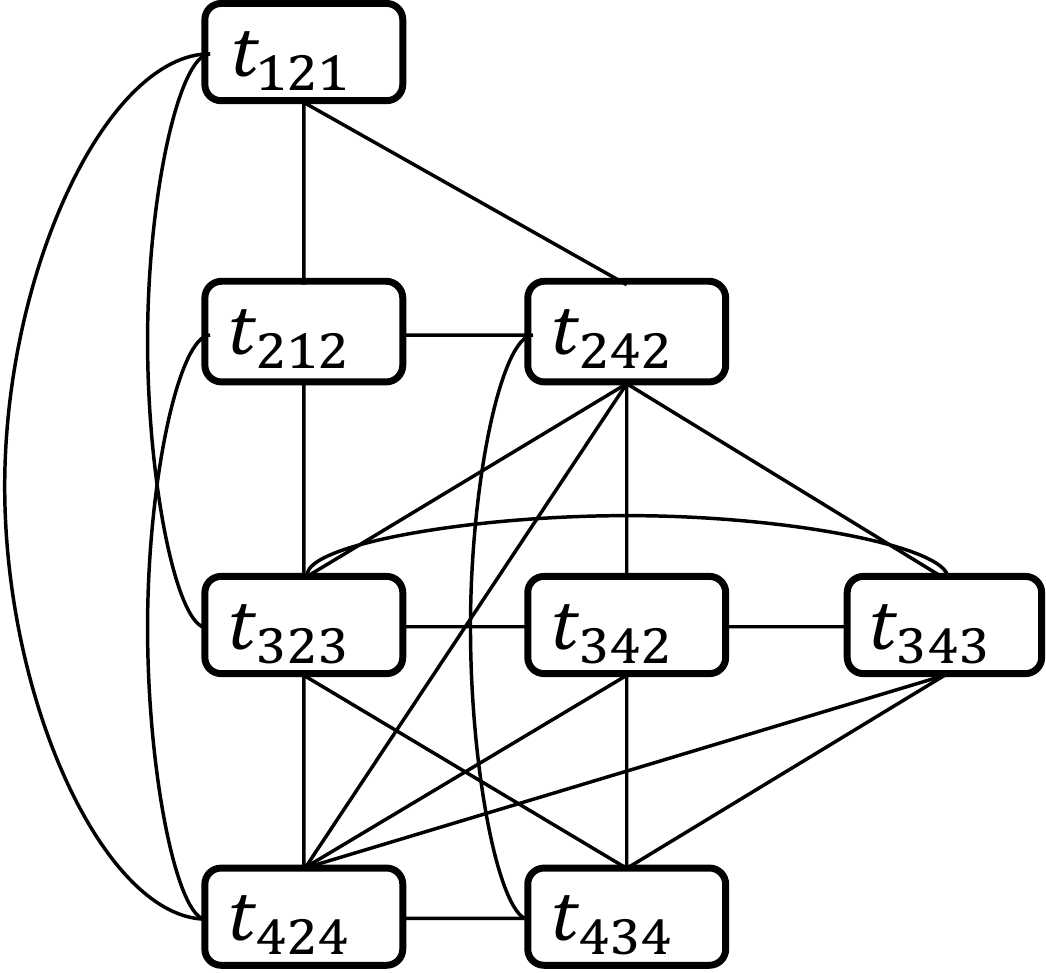}{}
    \label{fig:sched_graph}
  }
  \caption{Example of scheduling graph. $\op_{ijk}$ represents a transmission
          in which vehicle $v_i$ sends $d_k$ to vehicle $v_j$.}
  \label{fig:graph}
\end{figure}

For every time slot $\tau$,
the RSU selects transmitter and receiver vehicles from $\Vehicles$,
as well as data $d_k \in \ds_{i,\tau}$
to send for each transmitter vehicle $v_i$.
This selection is calculated by solving the MWIS problem
for the scheduling graphs, which are constructed as follows:

Algorithm~\ref{alg:sched_graph} is utilized
to construct scheduling graphs for each time slot $\SchedGraph[,\tau]\coloneqq\(\Operation_\tau,\Conflict_\tau,\weight\)$ from $\VehGraph$,
where $\Operation_\tau$, $\Conflict_\tau$, and $\weight$ denote
the set of vertices, set of edges,
and vertex weighting function such that $\weight\colon\Operation_\tau\to\mathbb{R}^{+}$, respectively.
$\mathbb{R}^{+}$ is the set of positive real numbers.
A transmission $\op_{ijk} \in \Operation_\tau$ 
represents a set containing transmitter $v_i$, receiver $v_j$, and data $d_k$,
meaning $v_i$ transmits $d_k$ to $v_j$.
Each element in $\Conflict_\tau$ represents
a conflict between two transmissions,
meaning they cannot be performed concurrently.
Further details are explained in Section \ref{sec:conflict}.
The weight of vertex $\weight(\op_{ijk})$ represents the priority of each transmission,
and its definition is described in Section \ref{sec:priority}.
Figure~\ref{fig:sched_graph} presents an example of a scheduling graph
constructed from the vehicular network graph shown in Fig.~\ref{fig:veh_graph}.
From lines 2--10 in Algorithm \ref{alg:sched_graph},
a set of transmissions $\Operation_\tau$ is obtained
by listing all directly connected pairs of vehicles
(i.e., neighbors in $\VehGraph$)
and data not possessed by receivers.
Next, the conflict between each pair of transmissions is calculated
in lines 11--15.

Algorithm~\ref{alg:scheduling} is a scheduling algorithm utilizing $\SchedGraph[,\tau]$.
In each time slot,
a set of transmissions $\IsOp_\tau \subset \Operation_\tau$ is selected.
After the transmissions are performed,
the datasets $\ds_{i,\tau+1}$ are updated utilizing (\ref{eq:updatedataset}),
and $\SchedGraph$ is recalculated based on the updated $\ds_{i,\tau+1}$.
We describe our method for selecting transmissions in the following subsection.
After the scheduling for all time slots in the data update interval is completed,
each vehicle follows the determined schedule during the data sharing period.

\subsection{Priority of Transmissions} \label{sec:priority}
To perform scheduling,
the controller calculates the MWIS of $\SchedGraph[,\tau]$
to increase the number of transmissions performed in each time slot.
Independent sets of $\SchedGraph[,\tau]$
represent sets of non-conflicting transmissions
and thus, maximum transmissions that can be performed concurrently
can be obtained by solving the maximum independent set (MIS) problem.
The MIS problem is a special case of MWIS,
where the weight function $\weight$ is a constant function
(i.e., $\weight(\op_{ijk}) = 1, \forall \op_{ijk} \in \Operation_\tau$).
We refer to the data sharing algorithm with MIS as max transmission scheduling.
Although max transmission scheduling
maximizes the number of transmissions,
it does not consider the perceivable region represented by the perceptual data.
When $\NumSlot$ is small due to the limit of the data sharing period,
the algorithm stops before all data are shared with all vehicles.
In such cases,
vehicles perceive their environments based on limited information
that covers only a limited area of the intersection.

To increase the perceivable area in such cases,
prioritization for transmitted data can be implemented.
Data from near the intersection tend to overlap with each other,
because the vehicle density near intersections
is higher than that far from intersections.
Therefore, it is inefficient to forward data representing areas near intersections.
We propose assigning a high priority to data that represent areas far from the intersection
and lower priority to data that represent areas near the intersection.
We refer to the algorithm with a priority function as max distance scheduling.
Distance priority can be represented
as a weight function $\weight$ of transmissions $\op_{ijk}$ as follows:
\begin{align}
  \weight(\op_{ijk}) &\coloneqq \dist(p_k, \centerpos),
\end{align}
where $p_k$, $\centerpos$, and $\dist(a,b)$ denote
the position of $v_k$,
center of the intersection,
and distance between positions $a$ and $b$,
respectively.

Although the MIS and MWIS problem are NP-hard problems,
it has been proven that a simple greedy algorithm
can approximately solve these problems
with a guaranteed performance ratio of greater than or equal to $1/\Delta$,
where $\Delta$ denotes the maximum degree of any vertex in the graph
\cite{noteMWIS}.
Thus, we adopt a greedy approach in our proposed scheduling algorithm.

\subsection{Conflict Rule of Interference} \label{sec:conflict}
When constructing $\SchedGraph[,\tau]$,
conflicts between transmissions $\op_{ijk} \in \Operation_\tau$ must be determined
to obtain the set of edges $\Conflict_\tau$.
The rules used to determine the conflicts are referred to as conflict rules.
The basic conflict rules are defined as follows:
\begin{itemize}
  \item[(a)] A transmitter cannot transmit different data at the same time
             or transmit data to different receivers because of the use of a narrow-beam directional antenna:
             $\op_{ijk}$ conflicts with $\op_{i'j'k'}$ if $i=i'$.
  \item[(b)] A receiver cannot receive data from multiple transmitters:
             $\op_{ijk}$ conflicts with $\op_{i'j'k'}$ if $j=j'$.
  \item[(c)] A vehicle cannot transmit and receive data simultaneously
             because of half-duplex communication:
             $\op_{ijk}$ conflicts with $\op_{i'j'k'}$ if $i=j' \lor j=i'$.
\end{itemize}
Because the original algorithm proposed in \cite{coopDataSched} assumed a DSRC channel and omnidirectional antennas,
the following conflict rule was added to the basic rules:
\begin{itemize}
  \item[(d)] A receiver near a transmitter cannot receive data from other transmitters:
             $\op_{ijk}$ conflicts with $\op_{i'j'k'}$ if $v_j \in \Neighbor(v_i') \lor v_j' \in \Neighbor(v_i)$,
             where $\Neighbor(v_i)$ denotes the neighbors of $v_i$.
\end{itemize}
However, this conflict rule does not match our problem because we assume mmWave V2V communications.

Considering the narrow beam width and high attenuation of mmWave communications,
it seems that the radio interference between two transmissions is negligible.
In this case, the conflict set $\Conflict_\tau$ is defined only by the basic rules: (a), (b), and (c).
However, interference sometimes occurs when an interferer is near a receiver
or the transmission direction of the desired and interfering signal are nearly parallel.

In order to overcome interference,
we design a conflict rule that reflects mmWave radio characteristics.
However, it is difficult to estimate SINR during scheduling
because interference cannot be calculated
before all the transmitters and their antenna directions are determined.
We propose an approximation method that can be adopted for our conflict rules,
which are defined for only two transmissions and utilized when constructing the scheduling graphs.

Considering the narrow beam width and high attenuation of mmWave radio signals,
we assume that the largest interference signal is the main factor of SINR.
Therefore, SINR can be approximated as follows:
\begin{align}
  \SINR_{i,j} &\coloneqq \frac{\PowerRx^{(i,j)}}{\BandWidth \Noise + \displaystyle{\sum_{k\neq i,j}\Itfr_k}} \\
          &\approx \frac{\PowerRx^{(i,j)}}{\BandWidth \Noise + \displaystyle{\max_{k\neq i,j} \Itfr_k}},
\end{align}
where $\SINR_{i,j}$, $\PowerRx^{(i,j)}$, and $\Itfr_k$ denote
the SINR at vehicle $v_i$, whose desired signal comes from $v_j$,
received signal strength of desired signals from $v_j$ at vehicle $v_i$,
and interference power from vehicle $v_k$, respectively.
Although knowledge regarding all interference signals seems to be required when calculating $\max_{k\neq i,j} \Itfr_k$,
a conflict rule can be designed between pairs of transmissions
by assuming that the currently considered interferer is the largest one.
The conflict rule reflecting interference is designed as follows:
\begin{itemize}
  \item[(d')] A receiver cannot receive data when interfering signals are large:
              $\op_{ijk}$ conflicts with $\op_{i'j'k'}$
              if $\sinr(\op_{ijk},\op_{i'j'k'}) \leq \Theta \lor \sinr(\op_{i'j'k'},\op_{ijk}) \leq \Theta$,
              where $\sinr(\op_{ijk},\op_{i'j'k'}) \coloneqq \PowerRx^{(j,i)}/(\BandWidth\Noise+\Itfr_{i'})$
              and $\Theta \coloneqq 2^{\Rate/\BandWidth} - 1$.
\end{itemize}
Consider the interfering signals from $v_{\jtx}$ and $v_{\ktx}$ to $v_{\irx}$,
where $\Itfr_{\jtx} < \Itfr_{\ktx}$.
$v_{\jtx}$ and $v_{\ktx}$ attempt to transmit signals to $v_{\jrx}$ and $v_{\krx}$, respectively,
and $v_{\irx}$ receives signals from $v_{\itx}$.
If $\sinr(\op_{\itx\irx a},\op_{\ktx\krx b}) \leq \Theta$,
then the concurrent transmission of $\op_{\itx\irx a}$ and $\op_{\ktx\krx b}$ cannot be scheduled by the proposed algorithm.
Therefore, when calculating the interference from $v_{\jtx}$ to $v_{\irx}$,
we do not need to consider interference from $v_{\ktx}$,
and the interference from $v_{\jtx}$ is assumed to be the largest.
This assumption can be extended inductively for more than three transmitters.

\subsection{Required Time Slot for Complete Data Sharing}
In this section,
we discuss the situation where sufficient time slots are available to complete data sharing.
First, we prove that 
the proposed scheduling algorithm terminates in finite time
and that all data are shared with all vehicles
if the data sharing time is not limited and the vehicular network graph $\VehGraph$ is connected.
We then discuss the bounds for the required number of time slots.

\begin{theorem}{}
  If $\VehGraph$ is connected,
  the proposed data sharing algorithm terminates in finite time
  and all the data initially possessed by vehicles
  are shared with all vehicles when the algorithm terminates,
  which can be expressed as follows:
  \begin{align}
    \forall i, \ds_{i,\tend} = \{d_1, \dots, d_{\Nveh} \}, \label{eq:data_end}
  \end{align}
  where $\tend$ denotes the step count at the end of Algorithm~\ref{alg:scheduling}.
\end{theorem}

\begin{proof}{}
  First, we prove that the proposed algorithm terminates in finite time
  and then, we prove that all vehicles possess all data at the end of the algorithm.

  Let $n_\tau$ denote the total size of the dataset $\ds_{i,\tau}$, defined as
  $n_\tau \coloneqq \sum_{i=1}^{\Nveh} |\ds_{i,\tau}|$,
  where $|\cdot|$ represents the cardinality of a set.
  When $\op_{ijk}$ is performed, meaning vehicle $v_j$ receives data $d_k$,
  $n_\tau$ is updated as $n_{\tau+1} \leftarrow n_\tau + 1$
  because $\Operation_\tau$ is constructed
  from all elements in $\op_{ijk}$ that satisfy $d_k \notin \ds_{j,\tau}$
  (lines 5--7 in Algorithm \ref{alg:sched_graph}).
  Let $m_\tau \coloneqq |\IsOp_\tau|$ denote the number of transmissions selected by the RSU.
  Then, $n_\tau$ is updated as $n_{\tau+1} \leftarrow n_\tau + m_\tau$ in each time slot.
  Meanwhile, the maximum value of $n_\tau$ is bounded by $\Nveh^2$.
  Therefore, the algorithm terminates in finite time
  if at least one transmission is selected in each time slot.

  Next, we prove that if there exists a vehicle
  that does not possess all data,
  at least one transmission can be performed.
  Assume $v_i$ does not possess $d_j$, which means $d_j \notin \ds_i$.
  Then, there exists a connected pair $\{v_\alpha, v_\beta\} \in \Link$ on the paths between $v_i$ and $v_j$
  that satisfies $d_j \in \ds_\alpha \land d_j \notin \ds_\beta$
  because at least $v_j$ possesses $d_j$.
  Note that the paths between $v_i$ and $v_j$ exist because $\VehGraph$ is connected.
  Now, we have $\op_{\alpha \beta j} \in \Operation_\tau$
  because $v_\alpha$ possesses $d_j$ and $v_\beta$ does not possess $d_j$,
  meaning $\Operation_\tau \neq \emptyset$.
  An independent set of a graph is not $\emptyset$
  if a vertex set of the graph is not $\emptyset$.
  Therefore, $\SchedGraph[,\tau]$ has an independent set whose size is greater than zero.

  If not all vehicles possess all data, a transmission can be performed
  and thus, the algorithm does not terminate.
  When the algorithm terminates, all vehicles possess all data.
  Because it is guaranteed that
  the algorithm always terminates in finite time,
  all data can be shared with all vehicles in finite time.
  \QED
\end{proof}

Next, we reveal the bounds of $\tend$.
\begin{corollary}{}
  If the vehicular network graph $\VehGraph$ is connected,
  $\tend$ is bounded as,
  \begin{align}
    \frac{\Nveh^2 - \Nveh}{\lfloor\Nveh/2\rfloor} \leq \tend \leq \Nveh^2 - \Nveh, \label{eq:bound}
  \end{align}
  where $\lfloor\cdot\rfloor$ represents the floor function.
\end{corollary}

\begin{proof}{}
  At the beginning of the algorithm, we have $n_0 = \Nveh$.
  From (\ref{eq:data_end}), we have $n_{\tend} = \Nveh^2$.
  Meanwhile, $n_{\tend}$ is also written as
  $n_{\tend} = n_0 + \sum_{\tau=0}^{\tend-1} m_\tau$.
  Because vehicles cannot transmit and receive data simultaneously,
  $m_\tau$ satisfies
  $m_\tau \leq \lfloor\Nveh/2\rfloor$.
  Additionally, $m_\tau$ also satisfies $m_\tau \geq 1$
  because at least one transmission is performed in each time slot.
  Therefore, we have
  $\frac{\Nveh^2 - \Nveh}{\lfloor\Nveh/2\rfloor} \leq \tend \leq \Nveh^2 - \Nveh$.
  \QED
\end{proof}

On one hand, $\tend$ is equal to the upper bound
if only one vehicle transmits data in every time slot.
On the other hand, $\tend$ achieves the lower bound
if half of the vehicles send data in every time slot,
which is the optimal case under the constraint
that each vehicle cannot send and receive data simultaneously.
In Section \ref{sec:results},
simulation results demonstrate that $\tend$ is near the lower bound in many cases.

\begin{table}[t]
  \caption{Simulation parameters}
  \label{tbl:simparams}
  \centering
  \begin{tabular}{cc}
    \hline
    Parameters & Values \\
    \hline
    Number of lanes & 4 \\
    Lane width & 3.5\,m \\
    Sidewalk width & 4\,m \\
    Sensor range $\SensorRange$ & 50\,m \\
    Bandwidth $\BandWidth$ & 2.16\,GHz \\
    Thermal noise $\Noise$ & -174\,dBm/Hz \\
    Data rate $\Rate$ & 1\,Gbit/s \\
    Transmission power $\PowerTx$ & 10\,dBm \\
    Antenna beam width & 15$^\circ$, 30$^\circ$ \\
    \hline
  \end{tabular}
\end{table}

\section{Simulation Results} \label{sec:results}
We evaluated our algorithm through simulations.
In our simulations,
the distribution of inter-vehicle distance followed an exponential distribution.
This assumption was confirmed in \cite{RoutingSparseVANET},
where the authors demonstrated that
the distribution of inter-vehicle distance
follows an exponential distribution
based on empirical data collected in a real environment.
We assumed the average distance between vehicles $\IntVDist$ in each lane
was approximately the same as the stopping distance for traffic safety.
We evaluated the situations where $\IntVDist=20\,\mathrm{m}$ and 40\,m,
which are slightly larger than the stopping distances
when the velocity is 40\,km/h and 60\,km/h, respectively \cite{world2008speed}.
We evaluated an intersection with two roads and four buildings assuming an urban area.
Each road had four lanes and sidewalks on both sides.
Each vehicle was modeled as a rectangle with a size of 1.7\,m $\times$ 4.4\,m.
The $\Nveh$ vehicles closest to the center of the intersection
shared their data with each other.
The number of participants $\Nveh$ was fixed to 20 and 40 when $\IntVDist$ was 20\,m
and $\Nveh$ was fixed to 10 and 20 when $\IntVDist$ was 40\,m.
If $\(\IntVDist, \Nveh\)=\DistNum{20}{20}$ and $\DistNum{40}{10}$,
the perceivable region with shared data covers $(25\,\mathrm{m}+\SensorRange/2)$ from the center of the intersection,
because vehicles within approximately 25\,m of the center of the intersection participate in data sharing.
When the number of vehicles is doubled,
perceivable region covers $(50\,\mathrm{m}+\SensorRange/2)$ from the center of the intersection.
The achievable coverage is sufficient for some applications,
e.g., accident or congestion detection system
with which a driver or self-driving system gets alerts
and stops or slows down the vehicle
when an accident or congestion is detected at the intersection.
We drew lines from the transmitter to the receiver vehicles
and counted the number of blocking vehicles on the lines.
The number of blockers was used to calculate the pass loss
based on the model proposed in \cite{PathLossPrediction}
and construct $\VehGraph$ from (\ref{eq:vehgraph}) and (\ref{eq:link}).
We also assumed that vehicles could not communicate with each other if the buildings blocked their line-of-sight path.
The antenna gain was calculated from the model in \cite{maltsev2010channel}.
The other parameters are listed in Table \ref{tbl:simparams}.

In our simulations, the RSU first determined the scheduling.
Next, vehicles transmitted data based on the scheduling.
When interference occurred, a receiver failed to receive data.
If a transmitter was scheduled to transmit data that it did not possess,
it did not transmit any data during that time slot.

\begin{figure}[t] \centering
  \includegraphics[width=0.41\textwidth]{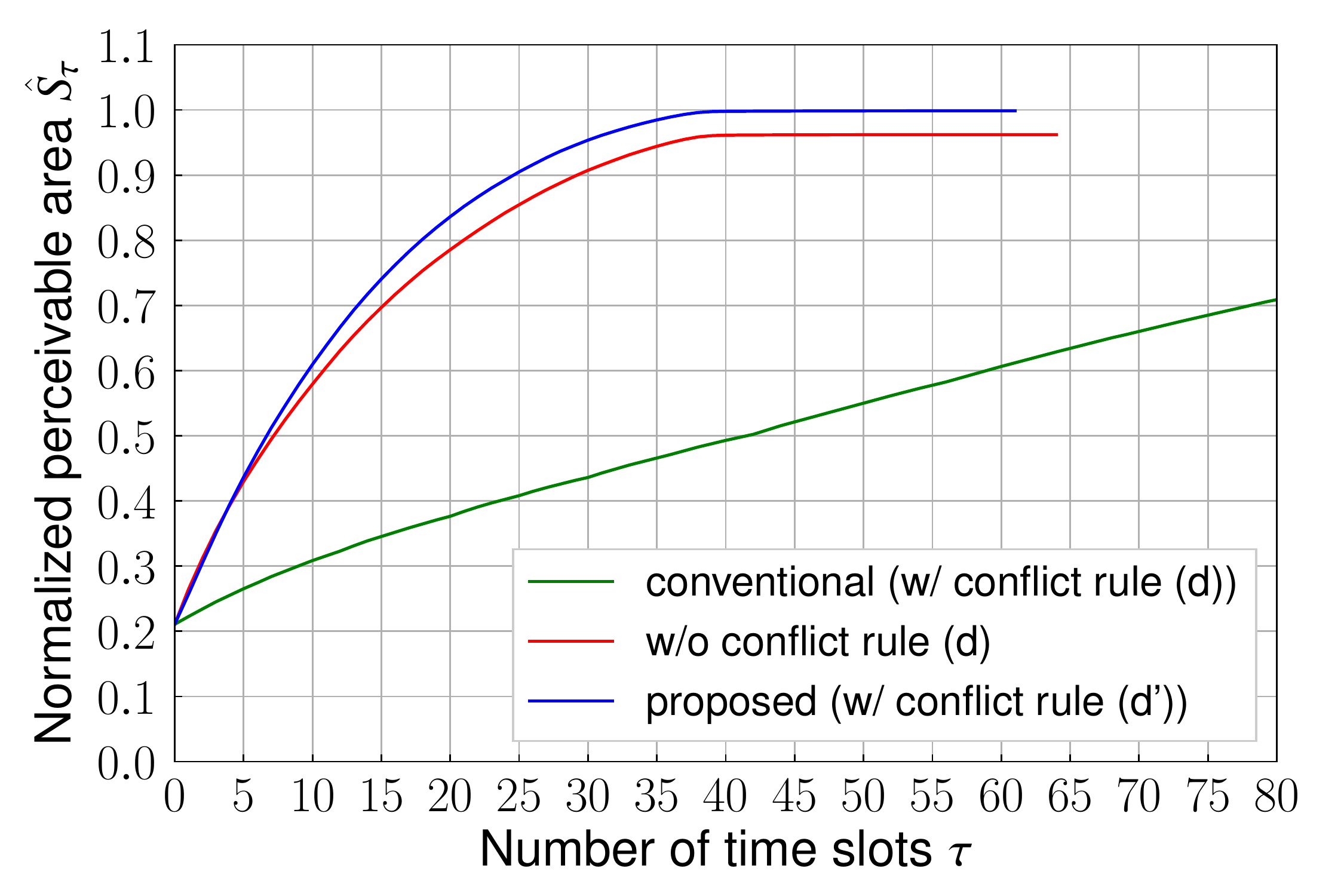}
  \caption{Normalized perceivable area as a function of the number of time slots
           when $\(\IntVDist,\Nveh\)=\DistNum{40}{20}$ and beam width is 15$^\circ$.
           The normalized perceivable area is enlarged by utilizing the conflict rule (d').}
  \label{fig:maxtrans}
\end{figure}
\begin{figure}[t] \centering
  \includegraphics[width=0.41\textwidth]{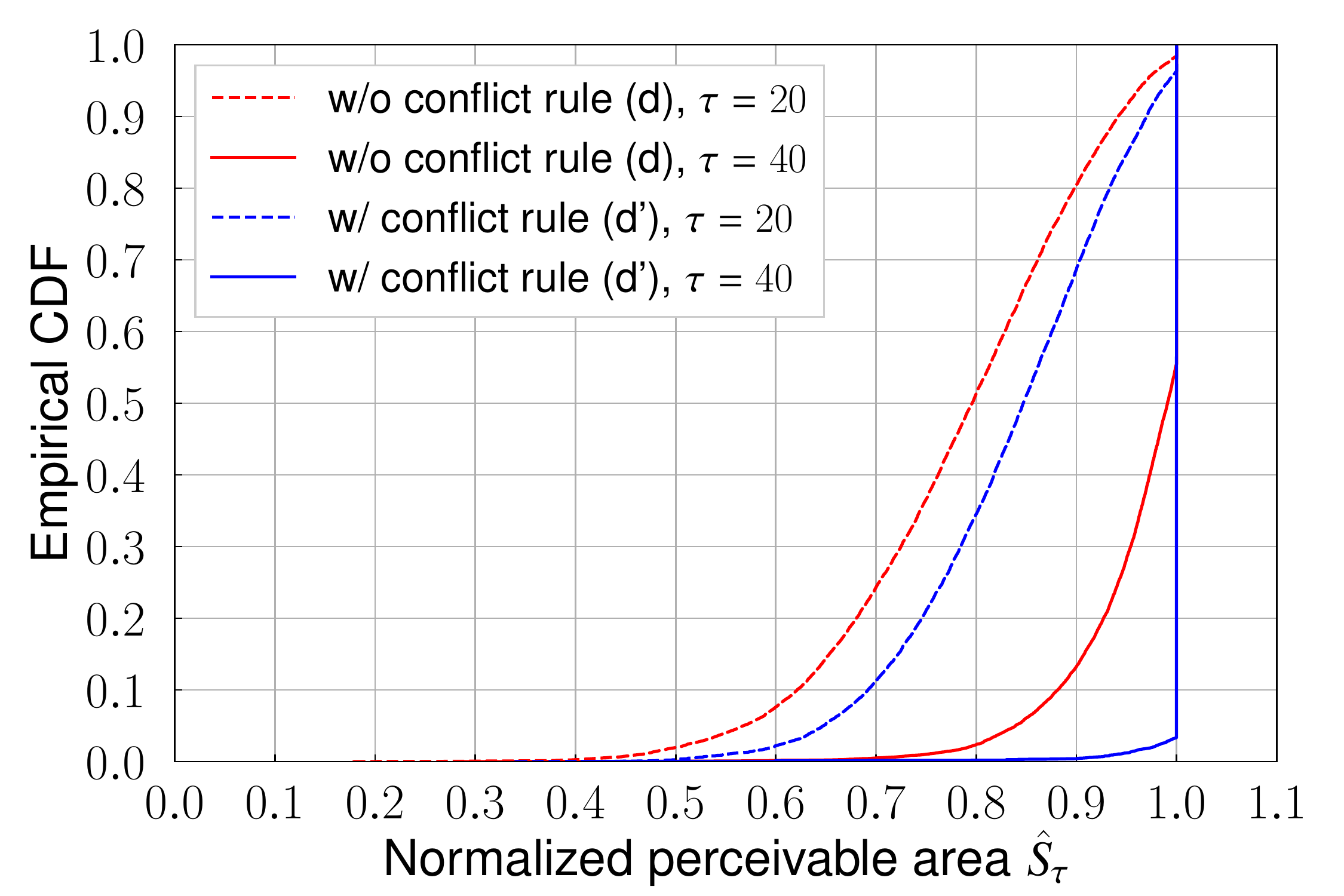}
  \caption{Empirical CDF of normalized perceivable area when $\(\IntVDist,\Nveh\)=\DistNum{40}{20}$.
           Nearly all vehicles achieve 90\% of the perceivable area at $\tau=40$
           when the conflict rule (d') is utilized.}
  \label{fig:cdf}
\end{figure}

Figure~\ref{fig:maxtrans} presents
the normalized perceivable area $\NormalCov[\tau]$ defined in Section \ref{sec:system}
as a function of the number of time slots $\tau$
when $\(\IntVDist,\Nveh\)=\DistNum{40}{20}$.
When the proposed data sharing algorithm was used,
the perceivable areas $\NormalCov[\tau]$ were enlarged by data sharing at first.
Then, $\NormalCov[\tau]$ saturated when $\tau\geq 40$
because the entire region $\RegionAll$ was covered by the shared perceptual data.
In other words, transmitted data after $\tau\geq 40$ did not contribute to enlarging the perceivable area because of overlap.
Finally, the algorithm terminated at $\tau=61$ when all scheduled transmissions were completed.
The proposed algorithm achieved
approximately twice the perceivable area compared with the conventional algorithm at $\tau=40$.
This is because the conventional algorithm assumed microwave communications
and thus, few vehicles could transmit data concurrently
because of the conflict rule (d), which was designed for microwave communications.
In contrast, the proposed method achieved efficient concurrent transmission
because its conflict rules reflect mmWave radio characteristics.
Additionally, adopting the conflict rule (d')
enlarged the perceivable area
because when this rule is not adopted,
certain interferences cannot be avoided and data sharing cannot be completed owing to transmission failure.

Figure~\ref{fig:cdf} shows the empirical cumulative distribution function (CDF)
of the normalized perceivable area $\NormalCov[\tau]$.
When utilizing the conflict rule (d'),
nearly all vehicles achieved 90\% of the normalized perceivable area at $\tau=40$,
whereas only 86\% of the vehicles achieved 90\% of the normalized perceivable area
where interference was not considered when determining the scheduling.

\begin{figure}[t] \centering
  \includegraphics[width=0.41\textwidth]{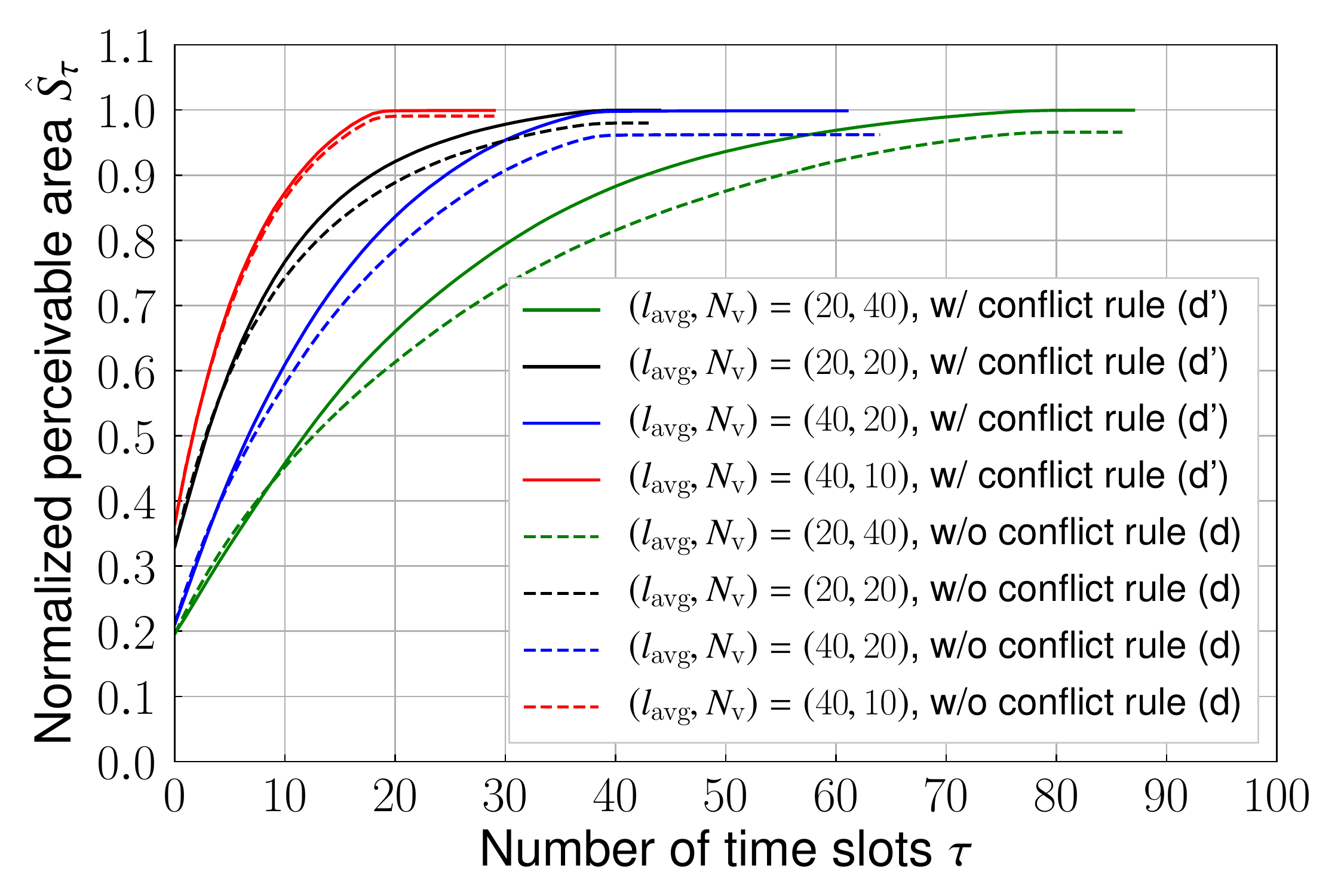}
  \caption{Normalized perceivable area when
           $\(\IntVDist,\Nveh\)=\DistNum{20}{40},$ $\DistNum{20}{20}, \DistNum{40}{20}, \DistNum{40}{10}$
           and the beam width is 15$^\circ$.
           The differences between the performances with and without the mmWave interference conflict rule are
           smaller when the number of participants is small or inter-vehicle distance is large.
           Averages of normalizing factors $\EachCov(\RegionAll)$ in (\ref{eq:nrm})
           for $\(\IntVDist,\Nveh\)=\DistNum{20}{40},\DistNum{20}{20},\DistNum{40}{20},$ and $\DistNum{40}{10}$
           are 6,161\,m$^2$, 3,949\,m$^2$, 5,793\,m$^2$, and 3,616\,m$^2$, respectively.
  }
  \label{fig:vnum}
\end{figure}
\begin{figure}[t] \centering
  \includegraphics[width=0.41\textwidth]{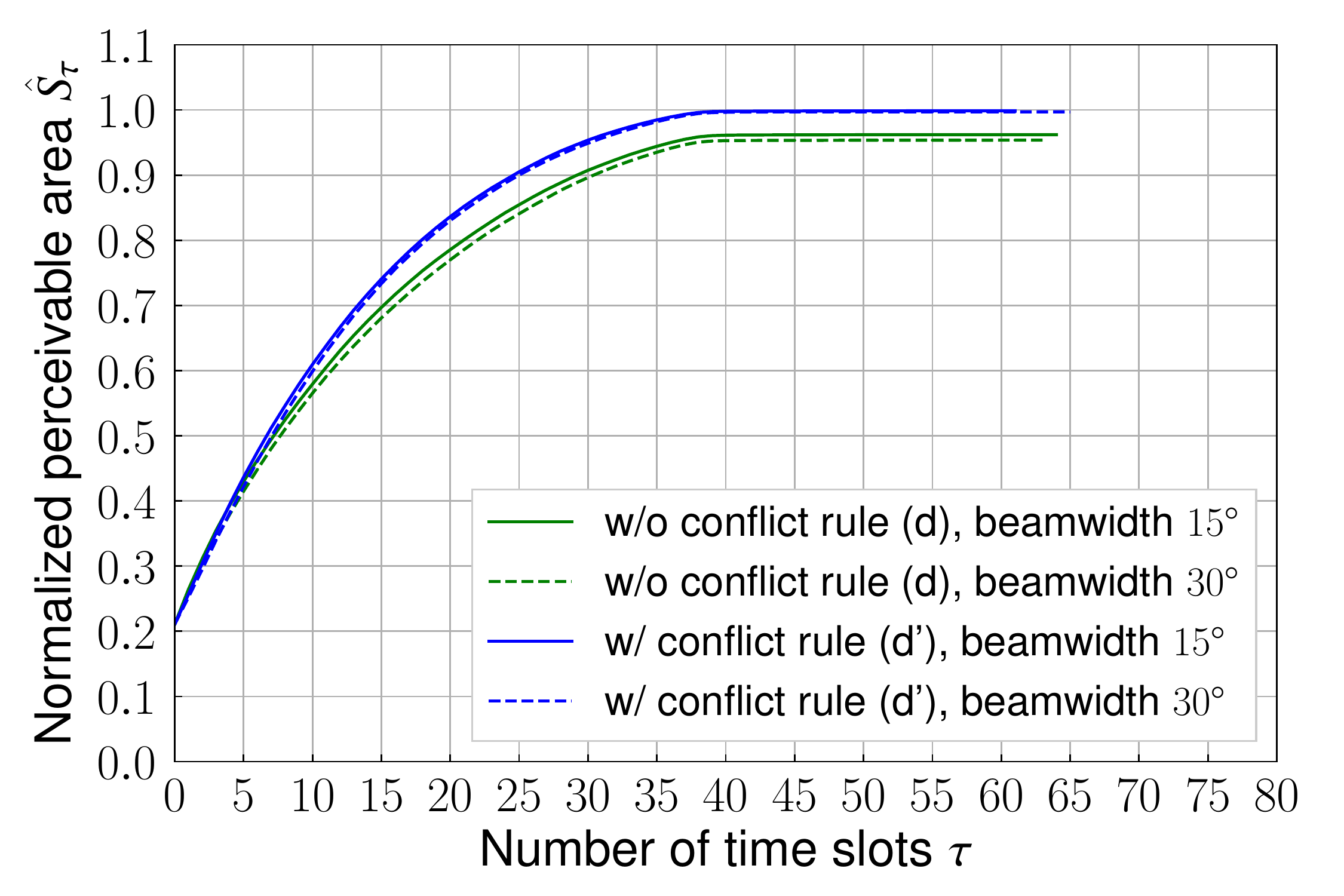}
  \caption{Normalized perceivable area when $\(\IntVDist,\Nveh\)=\DistNum{40}{20}$ and the beam width is 15$^\circ$ and 30$^\circ$.
           The differences between the beam widths of 15$^\circ$ and 30$^\circ$
           are larger when the mmWave interference conflict rule is not adopted
           compared with when the rule (d') is adopted.}
  \label{fig:bw}
\end{figure}

The normalized perceivable areas with different inter-vehicle distances and number of vehicles
are shown in Fig.~\ref{fig:vnum}.
The beam width was 15$^\circ$.
When the vehicle average inter-vehicle distance was the same,
the smaller the number of vehicles, 
the faster the data sharing terminated.
This is confirmed by (\ref{eq:bound}).
It is also shown that the superiority of adopting conflict rule (d')
does not depend on $\IntVDist$ and $\Nveh$.
The performance gain of adopting conflict rule (d') is larger
when the number of participants is large or the inter-vehicle distance is small
because interference is more likely to occur in these cases.

Figure~\ref{fig:bw} presents the normalized perceivable areas
with beam widths of 15$^\circ$ and 30$^\circ$,
when $\(\IntVDist, \Nveh\)=\DistNum{40}{20}$.
The differences between the beam widths of 15$^\circ$ and 30$^\circ$
were larger when the conflict rule for mmWave interference was not adopted
compared with when the rule (d') is adopted.
This is because interferences occur more frequently with a wider beam width,
but the proposed conflict rule (d') can successfully avoid interference.

Figure~\ref{fig:comp_weight} presents differences between the two priority designs
when $\(\IntVDist, \Nveh\)=\DistNum{40}{20}$.
The max distance scheduling design achieved a larger perceivable area
than the max transmission scheduling design when $\tau\leq 36$.
$\NormalCov[\tau]$ of the max distance design was 20\% larger than that of the max transmission design when $\tau=8$.
Based on the max distance scheduling,
data far from the center of the intersection were transmitted at first
and thus, the overlapping regions tended to be smaller than those in the algorithm without such priority control.
Therefore, $\NormalCov[\tau]$ became larger than that in the max transmission scheduling design.
When the number of time slot was limited, such that $\NumSlot \leq 36$,
the max distance scheduling design provided a larger perceivable area during every data update interval.
\begin{figure}[t] \centering
  \includegraphics[width=0.41\textwidth]{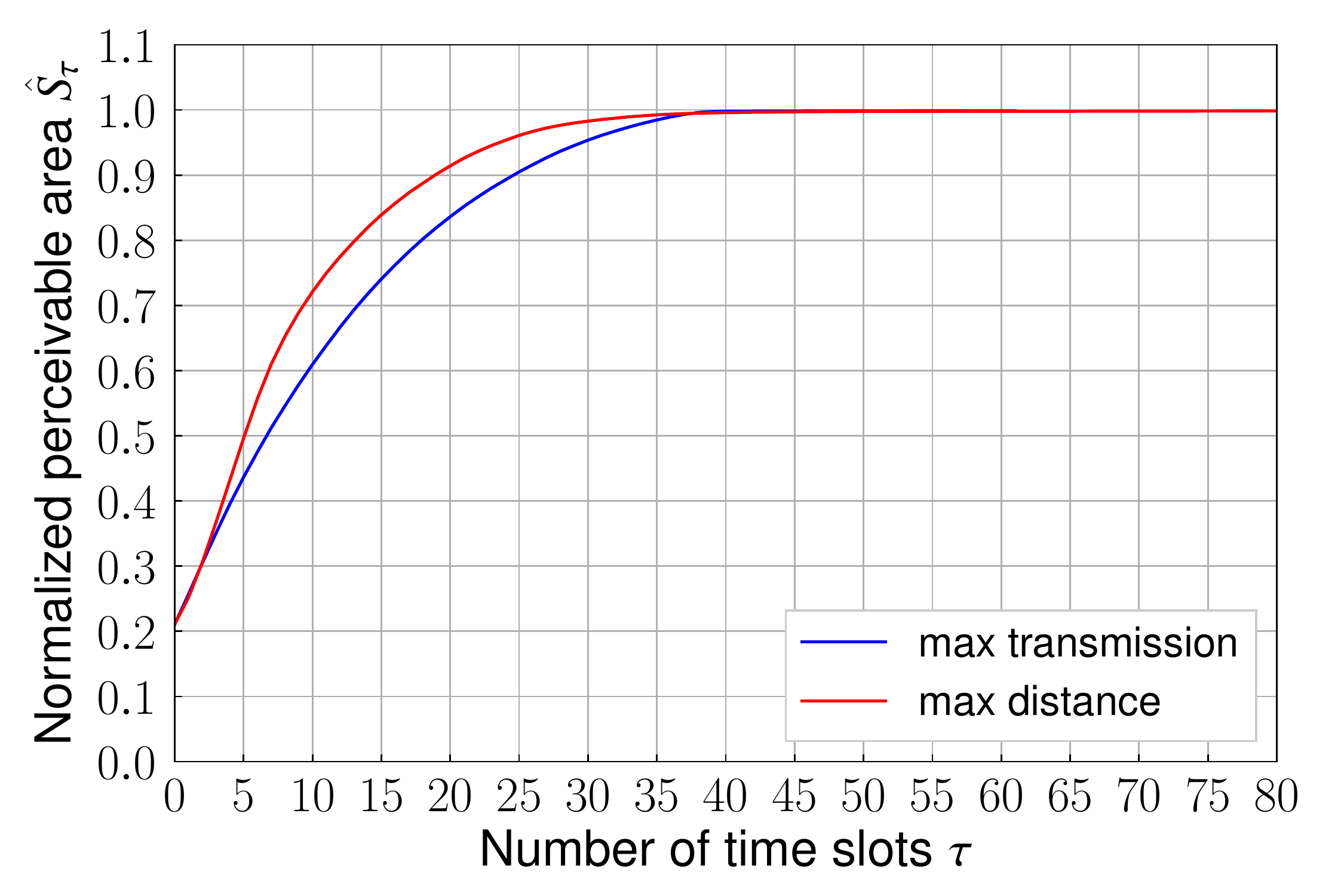}
  \caption{Normalized perceivable area with the different priority designs,
           when $\(\IntVDist,\Nveh\)=\DistNum{40}{20}$ and beam width is 15$^\circ$.
           The conflict rule (d') is adopted.
           The perceivable area can be enlarged by prioritizing data far from the intersection
           when the number of time slot is limited.}
  \label{fig:comp_weight}
\end{figure}

Figure~\ref{fig:step_cdf} shows empirical CDF
of the number of time slots required to share all data, denoted $\tend$,
when $\NumSlot$ is much larger than $\tend$.
Inter-vehicle distance was assumed to be 40\,m.
From (\ref{eq:bound}),
the lower bounds of $\tend$ were calculated as 18, 30, and 38,
for $\Nveh=10,15,20$.
The lower bounds are depicted as black vertical lines in Fig.~\ref{fig:step_cdf}.
In most cases, $\tend$ were closer to the lower bounds than the upper bounds of 90, 210, and 380.
In very few cases, as shown in Fig.~\ref{fig:step_cdf},
the data sharing algorithm terminated after fewer iterations than the lower bounds.
This is because
the vehicular network graphs $\VehGraph$ were disconnected in such cases
and thus, the data sharing algorithm terminated before all data were shared.
The differences between the protocols with and without prioritization
can be observed when $\Nveh=20$.
The max transmission scheduling design achieved efficient data sharing
because it maximized the number of transmitted data at each time slot,
meaning the algorithm terminated faster than the max distance scheduling design
when data sharing time was not limited.

\begin{figure}[t] \centering
  \includegraphics[width=0.41\textwidth]{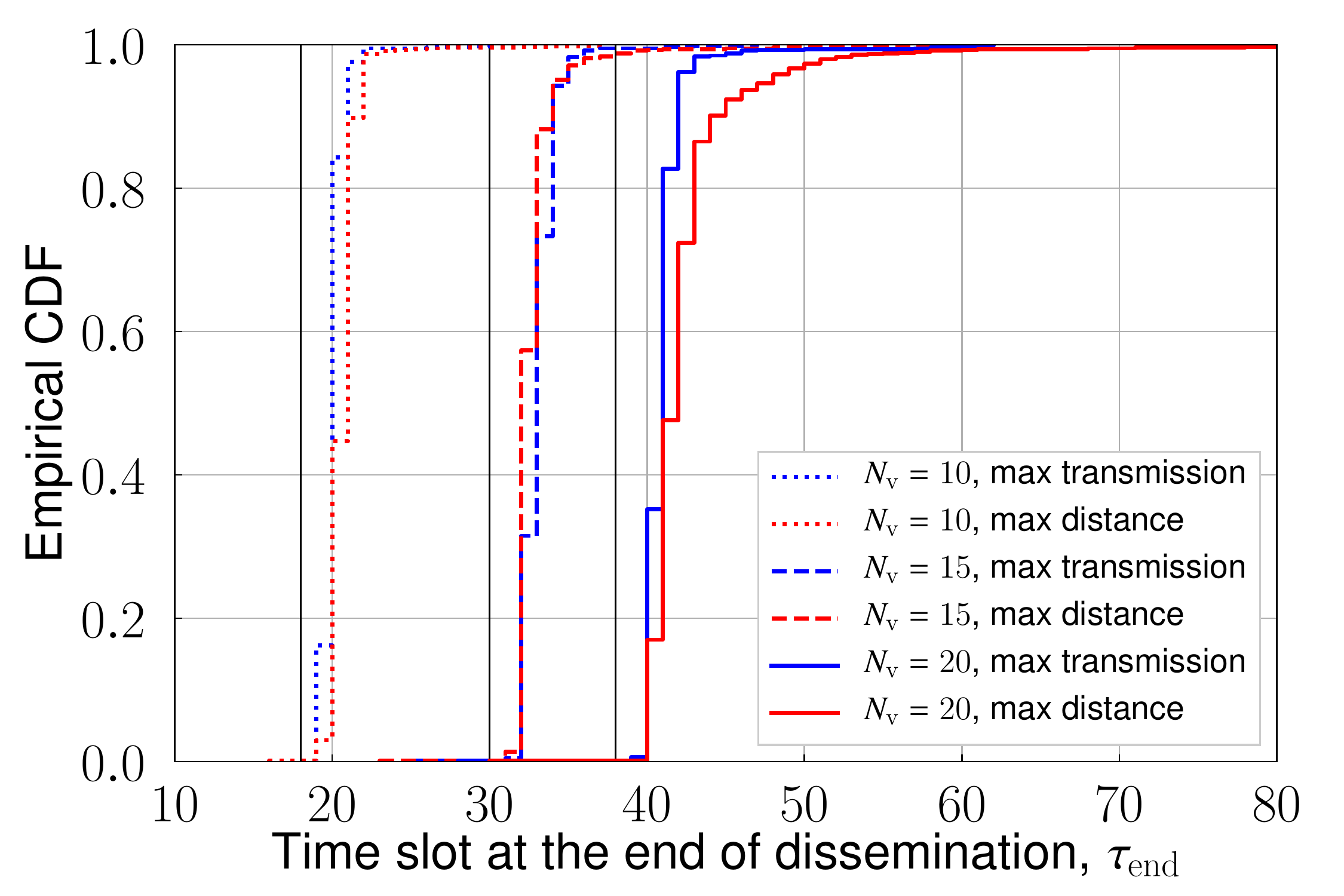}
  \caption{Empirical CDF of the number of time slots required to share all data.
           The conflict rule (d') is adopted.
           The beam width is 15$^\circ$
           and inter-vehicle distance is 40\,m.
           Black vertical lines represent lower bounds.
           The proposed algorithm achieves near-optimal scheduling.}
  \label{fig:step_cdf}
\end{figure}
\section{Conclusion} \label{sec:conclusion}
We proposed a data sharing scheduling method with concurrent transmission for mmWave VANETs for cooperative perception.
We modified the algorithm in \cite{coopDataSched}
by designing a conflict rule that represents mmWave communication characteristics
and a weight function that prioritizes data to be forwarded to enlarge the perceivable area.
Simulation results demonstrated that the proposed conflict rule for scheduling graphs
achieved a larger perceivable area compared with the original rules,
which did not consider directional antennas.
Priority control methods also enlarged the perceivable region
by sharing data that covered areas far from an intersection at first.
The priority control worked efficiently
in situations where the number of time slots was limited.
We also proved that
the proposed algorithms terminate in finite time
and all data can be shared with all vehicles
if a vehicular network is represented as a connected graph and there are sufficient time slots.

For future work,
we will develop a data-aggregation and vehicle-selection method for reducing redundant data transmissions.
The algorithm proposed in this paper transmits data
without considering overlapping regions covered by multiple data.
To suppress the transmission of data representing overlapping regions
can reduce data traffic without reducing the perceivable area.
The data-aggregation method that aggregates some overlapping data to a single datum
also reduces the amount of data to be transmitted.
When the vehicles densely located,
to select vehicles generating data considering their sensor coverage
can reduce data transmissions including the same regions.

Another interest is to develop a scalable distributed scheduling method.
In the proposed algorithm, an RSU, which act as a central controller,
determines the schedule based on information from all vehicles
participating in cooperative perception.
When the number of vehicles is large,
it is difficult for a central controller to obtain accurate information from all vehicles
and to perfectly control whole schedules
because mmWave channels vary rapidly.
Therefore, distributed scheduling including hybrid schemes of centralized and distributed scheduling
should be developed
to reduce the amount of vehicle information transmitted to the RSU
and to allow vehicles to decide scheduling autonomously using their own information.

\section*{Acknowledgment}
This work was supported in part by JSPS KAKENHI Grant Number JP17H03266,
KDDI Foundation, and Tateisi Science and Technology Foundation.

\bibliographystyle{ieicetr}
\bibliography{main}

\profile[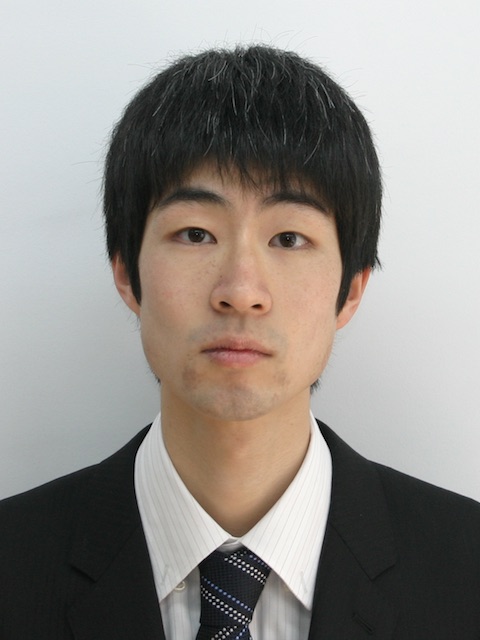]{Akihito Taya}{
received the B.E.\ degree in Electrical and Electronic Engineering from Kyoto University in 2011. 
He received the master degree in Communications and Computer Engineering, Graduate School of Informatics from Kyoto University, Kyoto, Japan, in 2013.
He joined Hitachi, Ltd. in 2013, where he perticipated in the development of computer clusters.
He is currently working toward a Ph.D.\ degree at the Graduate School of Informatics, Kyoto University.
His current research interests include vehicular communications and applications of machine learning.
He is a student member of the IEEE, ACM, and IEICE.
}
\profile[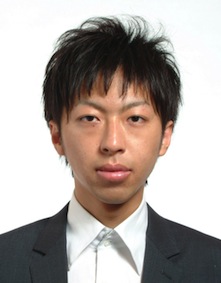]{Takayuki Nishio}{
received the B.E.\ degree in Electrical and Electronic Engineering from Kyoto University in 2010. 
He received the master and Ph.D. degrees in Communications and Computer Engineering, Graduate School of Informatics from Kyoto University, Kyoto, Japan, in 2012 and 2013, respectively.
From 2012 to 2013, he was a research fellow (DC1) of the Japan Society for the Promotion of Science (JSPS).
Since 2013, He is an Assistant Professor in Communications and Computer Engineering, Graduate School of Informatics, Kyoto University. From 2016 to 2017, he was a visiting researcher in Wireless Information Network Laboratory (WINLAB), Rutgers University, United States. 
His current research interests include mmWave networks, wireless local area networks, application of machine learning, and sensor fusion in wireless communications. He received IEEE Kansai Section Student Award in 2011, the Young Researcher's Award from the IEICE of Japan in 2016, and Funai Information Technology Award for Young Researchers in 2016. He is a member of the IEEE, ACM, IEICE.
}
\profile[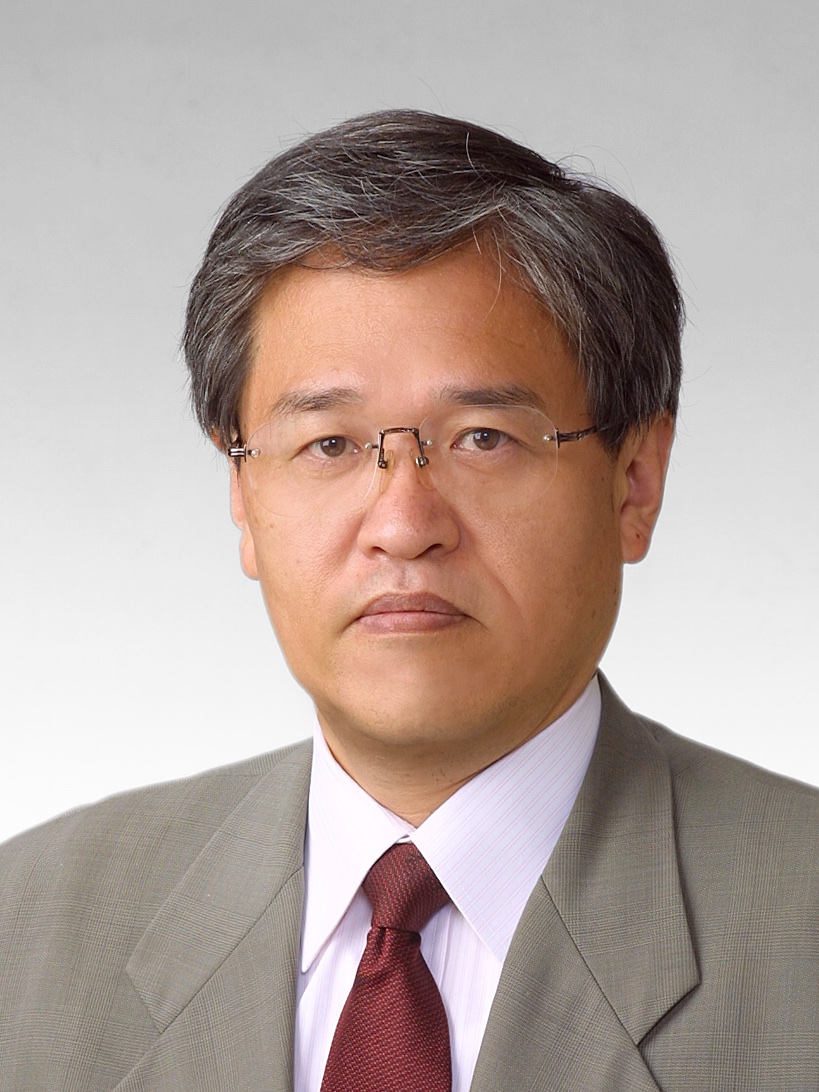]{Masahiro Morikura}{
Masahiro Morikura received his B.E., M.E., and Ph.D.\ degrees in electronics engineering from Kyoto University, Kyoto, Japan in 1979, 1981 and 1991, respectively.
He joined NTT in 1981, where he was engaged in the research and development of TDMA equipment for satellite communications. From 1988 to 1989,
he was with the Communications Research Centre, Canada, as a guest scientist.
From 1997 to 2002, he was active in the standardization of the IEEE 802.11a based wireless LAN.
His current research interests include WLANs and M2M wireless systems.
He received the Paper Award and the Achievement Award from IEICE in 2000 and 2006, respectively.
He also received the Education, Culture, Sports, Science and Technology Minister Award in 2007 and Maejima Award in 2008, and the Medal of Honor with Purple Ribbon from Japan's Cabinet Office in 2015.
Dr. Morikura is now a professor in the Graduate School of Informatics, Kyoto University. He is a member of the IEEE.
}
\pagebreak
\profile[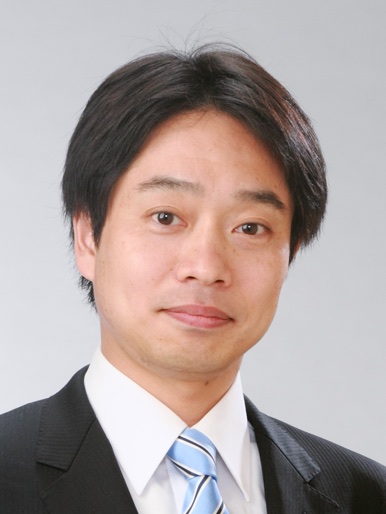]{Koji Yamamoto}{received the B.E.\ degree in electrical and electronic engineering from Kyoto University in 2002, and the M.E.\ and Ph.D.\ degrees in Informatics from Kyoto University in 2004 and 2005, respectively.
From 2004 to 2005, he was a research fellow of the Japan Society for the Promotion of Science (JSPS).
Since 2005, he has been with the Graduate School of Informatics, Kyoto University, where he is currently an associate professor.
From 2008 to 2009, he was a visiting researcher at Wireless@KTH, Royal Institute of Technology (KTH) in Sweden.
He serves as an editor of IEEE Wireless Communications Letters from 2017 and the Track Co-Chairs of APCC 2017 and CCNC 2018.
His research interests include radio resource management and applications of game theory.
He received the PIMRC 2004 Best Student Paper Award in 2004, the Ericsson Young Scientist Award in 2006.
He also received the Young Researcher's Award, the Paper Award, SUEMATSU-Yasuharu Award from the IEICE of Japan in 2008, 2011, and 2016, respectively, and IEEE Kansai Section GOLD Award in 2012.
He is a member of the IEEE.
}

\end{document}